\documentclass[letterpaper, 10 pt, conference,dvipsnames]{ieeeconf}  % Comment this line out
\IEEEoverridecommandlockouts                              % This command is only
\usepackage{andrews_macro}

\usepackage{comment}
\usepackage[colorinlistoftodos]{todonotes}  % Use option `disable' to make todo notes disappear
\usepackage{xargs} 

\newcommandx{\kam}[2][1=]{\todo[linecolor=ForestGreen!,backgroundcolor=ForestGreen!25,bordercolor=ForestGreen,#1]{[kam] #2}}
\newcommandx{\kmtodo}[1]{\textbf{\textcolor{ForestGreen}{TO-DO: #1}}}
\newcommandx{\edit}[2][1=]{\todo[linecolor=yellow!,backgroundcolor=yellow!25,bordercolor=yellow,#1]{[edit] #2}}
\newcommandx{\note}[2][1=]{\todo[linecolor=red!,backgroundcolor=red!25,bordercolor=red,#1]{[NOTE] #2}}
\makeatletter
\let\NAT@parse\undefined
\makeatother
\usepackage{hyperref}  %hyperref still needs to be put at the end!

\title{\LARGE \bf
Relative Pose Observability Analysis Using Dual Quaternions
}

\author{Nicholas B. Andrews and Kristi A. Morgansen% <-this % stops a space
\thanks{This work was supported by Blue Origin Enterprises, L.P.}% <-this % stops a space
\thanks{N.B. Andrews is a PhD student with the Department of Aeronautics and Astronautics, University of Washington, Seattle, WA 98195, USA {\tt\small nian6018@uw.edu}.}%
\thanks{K.A. Morgansen is a Professor with the Department of Aeronautics and Astronautics, University of Washington, Seattle, WA 98195, USA {\tt\small morgansn@uw.edu}.}%
}

\begin{document}

\maketitle
\thispagestyle{empty}
\pagestyle{empty}

% --- abstract
\begin{abstract}
Relative pose (position and orientation) estimation is an essential component of many robotics applications. Fiducial markers, such as the AprilTag visual fiducial system, yield a relative pose measurement from a single marker detection and provide a powerful tool for pose estimation. In this paper, we perform a Lie algebraic nonlinear observability analysis on a nonlinear dual quaternion system that is composed of a relative pose measurement model and a relative motion model. We prove that many common dual quaternion expressions yield Jacobian matrices with advantageous block structures and rank properties that are beneficial for analysis. We show that using a dual quaternion representation yields an observability matrix with a simple block triangular structure and satisfies the necessary full rank condition.
\end{abstract}

% --- intro
\section{INTRODUCTION}
In robotics applications, such as manipulation and cooperative control, the relative pose (position and orientation), angular velocity, and translational velocity between robots (or more generally, desired coordinate frames representing links, objects, etc.) are necessary for precision feedback control. The motivating scenario for the work is the use of fiducial markers for satellite relative proximity operations. In this use case, precise relative navigation is essential for mission success because measurements from Earth-based sensors may not be able to distinguish between satellites, and communication is often sparse and delayed. While each robot could use knowledge of its inertial pose and the inertial poses of the other relevant coordinate frames to calculate relative poses, tracking the relative poses directly is more convenient and computationally efficient \cite{Funda1990}, especially when working with more than a few coordinate frames simultaneously. Additionally, the measurement types often used in relative motion scenarios are also relative, which promotes the usage of a purely relative system model.

Of the options for relative coordinate representations, dual quaternions offer a compact framework with convenient properties for modeling six-degree-of-freedom systems that have been shown to be the most computationally efficient representation of rigid body transformations \cite{Funda1990}. Dual quaternions have been used to model the relative motions of manipulators and spacecraft \cite{Valverde2018-ri, Tsiotras2020-od}, to perform recursive state estimation \cite{Zivan2022-pw, Filipe2015-ep}, and to develop feedback controllers \cite{Filipe2015-zc, Stanfield2021-bt}. While the use of dual quaternions as a modeling framework is well documented, their application toward estimator design, controller design, and performance analysis is an active research area.

Pose estimation using image data as the sensing mechanism is an extensively researched topic in robotics. Two popular approaches for image-based pose estimation are machine learning methods, which require training on a bank of known objects, and extrinsic calibration, which uses a set of known feature points affixed to the observed object to reconstruct the relative pose \cite{Szeliski2022-gh}. In this paper, we investigate the latter approach with our work specifically inspired by the AprilTag visual fiducial system \cite{Wang2016-ae}. After preliminary camera calibration and using knowledge of the size of the fiducial marker, imaging a single fiducial marker provides the relative pose of the marker with respect to the camera.

Estimating the state of a system from measurements is well-studied in control theory with the viability of the estimation task determined by the system observability. The observability of a system characterizes the ability to uniquely determine the state from a set of measurements. In nonlinear systems, observability can be determined by taking appropriate Lie derivatives and constructing an observability matrix. If the observability matrix is full rank, then the system is locally weakly observable and the initial state can be uniquely determined from measurements \cite{Hermann1977-rt}. In this paper, we consider a nonlinear dual quaternion system that consists of a six degrees of freedom relative motion model and a relative pose measurement model.

While the overarching objective of this paper is a nonlinear observability analysis of the aforementioned system, the primary contribution is the development and demonstration of dual quaternions as an efficient analysis framework. To the best of our knowledge, this is the first time the nonlinear observability of a dual quaternion system has been investigated analytically. 
%The work presented in this paper builds on previous work presented in \cite{aas2023} which used the empirical observability Gramian for optimal fiducial marker placement on the surface of a satellite during relative proximity operations. 
In related work, the observability of the dual quaternion satellite relative motion model with line-of-sight measurements was empirically determined through Monte Carlo simulations and a Kalman filter covariance analysis \cite{Zivan2022-pw}. Observability conditions for a six-degree-of-freedom system with line-of-sight measurements were proven in \cite{Sun2002-pb} by investigating the rank of the error covariance matrix. In \cite{Li2019-ov}, observability conditions for a satellite relative motion model with relative position measurements and a purely translational state were derived using a Lie algebraic nonlinear observability approach. 

% In \cite{Zhou2008-me} the relative motion between planar three degrees of freedom robots with relative range measurements was shown to be locally weakly observable with at least three measurements and to be locally observable with five measurements.

The remaining sections are outlined as follows: in Section \ref{sec:quaternion}, we introduce quaternion and dual quaternion definitions and properties, while highlighting parallels between the two quaternion types. In Section \ref{sec:model}, the dual quaternion relative motion model is derived, and the relative pose measurement model is presented. Nonlinear observability and the observability matrix rank condition are discussed in Section \ref{sec:nlobsv}. In Section \ref{sec:obsv}, we take advantage of the block structure and rank properties of quaternion matrices derived in previous sections to prove that the nonlinear system is locally weakly observable. Lastly, future work is discussed in the conclusion.

% --- background
\section{QUATERNION FRAMEWORK} \label{sec:quaternion}
In this section, we provide the foundation of the quaternion framework which will be used extensively to derive the relative motion model and prove observability in subsequent sections. Quaternion and dual quaternion definitions and properties are summarized from \cite{Valverde2018-ri,Filipe2015-zc, Sola2017-mp}, and elementary operations are defined in Table \ref{tab:quaternion}.

\begin{table*}
    \vspace{0.2cm}
    \centering
    \makebox[\textwidth][c]{
    \begin{tabular}{|c | c | c|}
        \hline
        \textbf{Operation} & \textbf{Quaternion Definition} & \textbf{Dual Quaternion Definition} \\
        \hline
        Addition & $\qblank{a} + \qblank{b} = \mtx{\qscalar{a} + \qscalar{b}, \ \qvec{a} + \qvec{b}}$ & $\dualqblank{a} + \dualqblank{b} = (\qreal{a} + \qreal{b}) + \dualunit (\qdual{a} + \qdual{b})$ \\
        Scalar Multiplication & $\lambda \qblank{a} = \mtx{\lambda \qscalar{a}, \ \lambda \qvec{a}}$ & $\lambda \dualqblank{a} = (\lambda \qreal{a}) + \dualunit (\lambda \qdual{a})$ \\
        Multiplication & $\qblank{a} \qblank{b} = \mtx{ \qscalar{a} \qscalar{b} - \qvec{a} \qdot \qvec{b}, \ \qscalar{a} \qvec{b} + \qscalar{b} \qvec{a} + \qvec{a} \cross \qvec{b}}$ & $\dualqblank{a} \dualqblank{b} = (\qreal{a} \qreal{b}) + \dualunit (\qdual{a} \qreal{b} + \qreal{a} \qdual{b})$ \\
        Conjugate & $\qconj{\qblank{a}} = \mtx{\qscalar{a}, \ -\qvec{a}}$ & $\qconj{\dualqblank{a}} = (\qconj{\qreal{a}}) + \dualunit (\qconj{\qdual{a}})$ \\
        Dot Product & $\qblank{a} \qdot \qblank{b} = \mtx{\qscalar{a} \qscalar{b} + \qvec{a} \qdot \qvec{b}, \ \zeros{3}{1} }$ & $\dualqblank{a} \qdot \dualqblank{b} = (\qreal{a} \qdot \qreal{b}) + \dualunit (\qdual{a} \qdot \qreal{b} + \qreal{a} \qdot \qdual{b})$ \\
        Cross Product & $\qblank{a} \cross \qblank{b} = \mtx{0, \ \qscalar{a} \qvec{b} + \qscalar{b} \qvec{a} + \qvec{a} \cross \qvec{b}}$ & $\dualqblank{a} \cross \dualqblank{b} = (\qreal{a} \cross \qreal{b}) + \dualunit (\qdual{a} \cross \qreal{b} + \qreal{a} \cross \qdual{b})$ \\
        Norm & $\norm{\qblank{a}}^2 = \qblank{a} \qdot \qblank{a} = \qconj{\qblank{a}} \qblank{a}$ & $\norm{\dualqblank{a}}^2 = (\qreal{a} \qdot \qreal{a} + \qdual{a} \qdot \qdual{a}) + \dualunit 0$ \\
        Swap & Undefined & $\swap{\dualqblank{a}} = \qdual{a} + \dualunit \qreal{a}$ \\
        Real & Undefined &  $\qreal{(\dualqblank{a})} = \qreal{\qblank{a}}$ \\
        Dual & Undefined &  $\qdual{(\dualqblank{a})} = \qdual{\qblank{a}}$ \\
        \hline
    \end{tabular}
    }
    \caption{Quaternion and dual quaternion operations \cite{Filipe2015-zc, Stanfield2021-bt}.}
    \label{tab:quaternion}
    \vspace{-1cm}
\end{table*}

\subsection{Quaternion}
A quaternion is defined as:
\begin{align}
    \qblank{q} = \qscalar{q} + \qblank{q_1} i + \qblank{q_2} j + \qblank{q_3} k.
\end{align}
where $i^2 = j^2 = k^2 = -1$, $i = jk = -kj$, $j = ki = -ik$, $k = ij = -ji$, and $ijk = -1$. A quaternion is composed of a scalar part, $\qscalar{q} \in \reals{}$, and vector part, $\qvec{q} =  {\left[ \qblank{q_1}, \qblank{q_2}, \qblank{q_3} \right]} \in \reals{3}$, succinctly expressed as the ordered pair $\qblank{q} = {\left[ \qscalar{q}, \qvec{q} \right]}$. In some literature, the scalar component is the last entry in the quaternion array, however, for this paper the scalar component is always the first entry. The set of quaternions, scalar quaternions, and vector quaternions are defined as $\quats = \{\qblank{q} \st  \qscalar{q} + \qblank{q_1} i + \qblank{q_2} j + \qblank{q_3} k, \qscalar{q} \in \reals{}, \qblank{q_1} \in \reals{}, \qblank{q_2} \in \reals{}, \qblank{q_3} \in \reals{} \}$, $\quatss = \{\qblank{q} \in \quats \st \qblank{q_1} = \qblank{q_2} = \qblank{q_3} = 0 \}$, and $\quatsv = \{\qblank{q} \in \quats \st \qscalar{q} = 0 \}$.

The relative orientation of a frame $\cf{x}$ with respect to a frame $\cf{y}$ is represented by the \emph{unit quaternion}, $\q{x}{y} \in \quatsu$. The set of unit quaternions is defined as $\quatsu = \{ \quats \st  \qconj{\qblank{q}} \qblank{q} = \qblank{q} \qconj{\qblank{q}} = \qblank{q} \qdot \qblank{q} = \qone \}$, where $\qone = { \left[1, \zeros{3}{1} \right]} \in \quatss$. The primary advantage of the unit quaternion attitude representation compared to Euler angles or alternative representations is that quaternions are singularity-free. A unit quaternion also has the following inverse properties:
\begin{gather}
    \label{eq:unitq}
    \inv{\q{x}{y}} = \qconj{\q{x}{y}} = \q{y}{x}.
\end{gather}
A unit quaternion can be expressed as a rotation angle, $\phi$, about a unit vector, $\ivec{n}$:
\begin{align} \label{eq:qdef}
    \q{x}{y} = {\left[ \cos \left( \frac{\phi}{2} \right), \ivec{n} \sin \left( \frac{\phi}{2} \right) \right]}.
\end{align}

Unit quaternions have a convenient form for transforming a vector between coordinate frames. By representing a vector $\qvecf{v}{x} \in \reals{3}$ in frame $\cf{x}$ coordinates as a vector quaternion $\qblankf{v}{x} = { \left[ 0, \qvecf{v}{x} \right] } \in \quatsv$, a coordinate transformation to and from the $\cf{y}$ frame has the form
\begin{gather} \label{eq:quatrot}
    \qblankf{v}{y} = \qconj{\q{y}{x}} \qblankf{v}{x} \q{y}{x}, \quad
    \qblankf{v}{x} = \q{y}{x} \qblankf{v}{y} \qconj{\q{y}{x}}.
\end{gather}
Additionally, unit quaternions can be chained together to solve for the total relative rotation between multiple reference frames: $\q{x}{y} = \qconj{\q{y}{z}} \q{x}{z}$.

The unit quaternion kinematic equations are
\begin{align}
    \label{eq:qkin}
    \dq{x}{y} = \frac{1}{2} \q{x}{y} \qomega{x}{y}{x} = \frac{1}{2} \qomega{x}{y}{y} \q{x}{y},
\end{align}
where $\omeg{x}{y}{x} \in \reals{3}$ is the angular velocity of $\cf{x}$ relative to $\cf{y}$ expressed in $\cf{x}$ coordinates, and $\qomega{x}{y}{x} = { \left[0, \omeg{x}{y}{x} \right]} \in \quatsv$.

\subsubsection{Vector Form}
Multiplication and conjugation of quaternions can be rewritten in a linear algebraic form allowing for easier manipulation and use of matrix calculus tools for quaternion calculus. A quaternion $\qblank{q} \in \quats$ left multiplied by a matrix $M \in \reals{4 \times 4}$ follows the standard matrix multiplication algebra:
\begin{gather}
    {M} = \mtx{{M}_{11} & {M}_{12} \\ {M}_{21} & {M}_{22}} \in \reals{4 \times 4} \nonumber \\
    {M}_{11} \in \reals{}, {M}_{12} \in \reals{1 \times 3}, {M}_{21} \in \reals{3 \times 1}, {M}_{22} \in \reals{3 \times 3} \\
    {M} \qblank{q} = \mtx{{M}_{11} \qscalar{q} + {M}_{12} \qvec{q}, \ {M}_{21} \qscalar{q} + {M}_{22} \qvec{q}} \in \quats. \nonumber
\end{gather}

The left and right quaternion multiplication matrices for a quaternion $\qblank{q} \in \quats$ are:
\begin{gather}
    \skewmat{q} = \mtx{0 & -\qblank{q}_3 & \qblank{q}_2 \\
    \qblank{q}_3 & 0 & -\qblank{q}_1 \\
    -\qblank{q}_2 & \qblank{q}_1 & 0} \in \reals{3 \times 3} \\
    \qleft{q} = \qscalar{q} \eye{4} + \mtx{0 & -\tpose{\qvec{q}} \\
    \qvec{q} & \skewmat{q}} \in \reals{4 \times 4} \\
    \qright{q} = \qscalar{q} \eye{4} + \mtx{0 & -\tpose{\qvec{q}} \\
    \qvec{q} & -\skewmat{q}} \in \reals{4 \times 4}
\end{gather}
where $\eye{n}$ is the $n \times n$ identity matrix. Quaternion multiplication can then be re-written as
\begin{gather}
    \qblank{a} \qblank{b} = \qleft{a} \qblank{b} = \qright{b} \qblank{a} \\
    \qblank{a} \qblank{b} \qblank{c} = \qleft{\qblank{a} \qblank{b}} \qblank{c} = \qright{\qblank{b} \qblank{c}} \qblank{a} = \left( \qright{\qblank{b}} \qblank{a} \right) \qblank{c}
\end{gather}
for $\qblank{a}, \qblank{b}, \qblank{c} \in \quats$.

The quaternion conjugate distributes like the matrix transpose and can be deconstructed using the conjugate matrix $\eyeconj$:
\begin{gather}
    \eyeconj = \diag{1, -1, -1, -1} \\
    \qconj{\qblank{a}} = \eyeconj \qblank{a}, \quad
    \qconj{\left( \qblank{a} \qblank{b} \right)} = \qconj{\qblank{b}} \qconj{\qblank{a}}.
\end{gather}

% Lastly, for a \emph{vector quaternion} $\qblank{r}$ with scalar component $\qscalar{r} = 0$:
% \begin{gather}
%     \qblank{r} \qblank{r} = -\norm{r}^2
% \end{gather}

\subsubsection{Derivatives}
By rewriting quaternion expressions in their linear algebraic form, taking their derivatives becomes much simpler because matrix calculus tools can be readily applied. Some common quaternion derivatives that will be used to construct dual quaternion derivatives in the subsequent subsection are:
\begin{gather}
    \pd{\norm{\qblank{q}}^2}{\qblank{q}} = 2 \qblank{q}, \quad
    \pd{\qblank{a} \qblank{b}}{\qblank{a}} = \qright{b}, \quad
    \pd{\qblank{a} \qblank{b}}{\qblank{b}} = \qleft{a} \\
    \pd{\qconj{\qblank{a}} \qblank{b} \qblank{a}}{\qblank{a}} = \qleft{\qconj{\qblank{a}} \qblank{b}} + \qright{\qblank{b} \qblank{a}} \eyeconj.
\end{gather}

\subsection{Dual Quaternion}
Dual quaternions are an extension of quaternions and provide a convenient and natural form for modeling the relative pose and velocities between coordinate frames. Similar to how a complex number is composed of a real and imaginary part, a dual quaternion is formed by a real part, $\qreal{q} \in \quats$, and a dual part, $\qdual{q} \in \quats$. The dual unit, $\dualunit$, has the properties $\dualunit^2 = 0$ and $\dualunit \neq 0$. The set of dual quaternions, scalar dual quaternions, and vector dual quaternions are defined as $\dquats = \{\dualqblank{q} \st \dualqblank{q} = \qreal{q} + \dualunit \qdual{q}, \qreal{q} \in \quats, \qdual{q} \in \quats \}$, $\dquatss = \{\dualqblank{q} \st \dualqblank{q} = \qreal{q} + \dualunit \qdual{q}, \qreal{q} \in \quatss, \qdual{q} \in \quatss \}$, and $\dquatsv = \{\dualqblank{q} \st \dualqblank{q} = \qreal{q} + \dualunit \qdual{q}, \qreal{q} \in \quatsv, \qdual{q} \in \quatsv \}$.

The pose of a frame $\cf{x}$ with respect to a frame $\cf{y}$ is represented by the \emph{dual pose}, $\dualq{x}{y} \in \dquatsu$, and is defined as
\begin{align}
    \dualq{x}{y} &= \q{x}{y} + \dualunit \frac{1}{2} \qpos{x}{y}{y} \q{x}{y} \\
    &= \q{x}{y} + \dualunit \frac{1}{2} \q{x}{y} \qpos{x}{y}{x} 
\end{align}
where $\qpos{x}{y}{x} = { \left[ 0, \pos{x}{y}{x} \right] } \in \quatsv$, $\pos{x}{y}{x} \in \reals{3}$ is the position of $\cf{x}$ relative to $\cf{y}$ expressed in $\cf{x}$ coordinates, and $\q{x}{y} \in \quatsu$ is the orientation of the $\cf{x}$ frame with respect to the $\cf{y}$ frame. The dual pose belongs to the set of \emph{unit dual quaternions} which are defined as $\dquatsu = \{\dualqblank{q} \in \dquats \st \qconj{\dualqblank{q}} \dualqblank{q} = \dualqblank{q} \qconj{\dualqblank{q}} = \dualqblank{q} \qdot \dualqblank{q} = \dualone\}$, where $\dualone = \qone + \dualunit \qzero \in \dquatss$ and $\qzero = { \left[0, \zeros{3}{1}\right]} \in \quatss$.

Many parallels exist between quaternions and dual quaternions in terms of properties and expression forms for changing coordinate frames, kinematics, and derivatives. Like the unit quaternion properties in \eqref{eq:unitq}, unit dual quaternions have similar inverse properties:
\begin{gather}
    \inv{\dualq{x}{y}} = \qconj{\dualq{x}{y}} = \dualq{y}{x}.
\end{gather}
Changing coordinate frames between vector dual quaternions $\dualqblankf{\omega}{x}, \dualqblankf{\omega}{y} \in \dquatsv$ also has a form similar to vector quaternions:
\begin{gather}
    \dualqblankf{\omega}{y} = \qconj{\dualq{y}{x}} \dualqblankf{\omega}{x} \dualq{y}{x}, \quad
    \dualqblankf{\omega}{x} = \dualq{y}{x} \dualqblankf{\omega}{y} \qconj{\dualq{y}{x}}.
\end{gather}
Additionally, unit dual quaternions can be chained together over intermediate coordinate frames to solve for the total relative transformation between frames: $\dualq{x}{y} = \qconj{\dualq{y}{z}} \dualq{x}{z}$.

The \emph{dual velocity} is a vector dual quaternion and embeds the relative rotational and translational velocities. It has the form
\begin{align}
    \dualomega{x}{y}{z} = \qomega{x}{y}{z} + \dualunit (\qvel{x}{y}{z} + \qomega{x}{y}{z} \cross \qpos{z}{x}{z}) \in \dquatsv,
\end{align}
where $\qomega{x}{y}{z} = { \left[ 0, \omeg{x}{y}{z} \right] } \in \quatsv$, $\omeg{x}{y}{z} \in \reals{3}$ is the angular velocity of $\cf{x}$ relative to $\cf{y}$ expressed in $\cf{z}$ coordinates,  $\qvel{x}{y}{z} = { \left[ 0, \vel{x}{y}{z} \right] } \in \quatsv$, and $\vel{x}{y}{z} \in \reals{3}$ is the translational velocity of $\cf{x}$ relative to $\cf{y}$ expressed in $\cf{z}$ coordinates. Calculating relative dual velocities is done similarly to vectors: $\dualomega{x}{y}{z} = \dualomega{x}{w}{z} - \dualomega{y}{w}{z}.$

The dual quaternion kinematics in \eqref{eq:dqkin} have a similar form to the quaternion kinematics in \eqref{eq:qkin}. Despite the familiar-looking form, it is important to remember that the dual quaternion kinematics capture the kinematics of the full pose:
\begin{align}
    \label{eq:dqkin}
    \ddualq{x}{y} = \frac{1}{2} \dualq{x}{y} \dualomega{x}{y}{x} = \frac{1}{2} \dualomega{x}{y}{y} \dualq{x}{y}.
\end{align}

\subsubsection{Vector Form}
Dual quaternion operations can also be expressed in a linear algebraic vector form, where a dual quaternion $\dualqblank{q} \in \dquats$ left multiplied by a matrix follows the standard matrix multiplication algebra:
\begin{gather}
    \hat{M} = \mtx{\hat{M}_{11} & \hat{M}_{12} \\ \hat{M}_{21} & \hat{M}_{22}} \in \reals{8 \times 8} \nonumber \\
    \hat{M}_{11}, \hat{M}_{12}, \hat{M}_{21}, \hat{M}_{22} \in \reals{4 \times 4} \\
    \hat{M} \dualqblank{q} = (\hat{M}_{11} \qreal{q} + \hat{M}_{12} \qdual{q}) + \dualunit (\hat{M}_{21} \qreal{q} + \hat{M}_{22} \qdual{q}) \in \dquats. \nonumber
\end{gather}

For a dual quaternion $\dualqblank{q} \in \dquats$, the left and right dual quaternion multiplication matrices in $\reals{8 \times 8}$ are:
\begin{gather}
    \qleft{\dualqblank{q}} = \mtx{\qleft{\qreal{\qblank{q}}} & 0 \\
    \qleft{\qdual{\qblank{q}}} & \qleft{\qreal{\qblank{q}}}}, \quad
    \qright{\dualqblank{q}} = \mtx{\qright{\qreal{\qblank{q}}} & 0 \\
    \qright{\qdual{\qblank{q}}} & \qright{\qreal{\qblank{q}}}} \label{eq:dqleftright}.
\end{gather}
Dual quaternion multiplication can be expressed in linear algebraic form as
\begin{gather}
    \dualqblank{a} \dualqblank{b} = \qleft{a} \dualqblank{b} = \qright{b} \dualqblank{a} \\
    \dualqblank{a} \dualqblank{b} \dualqblank{c} = \qleft{\dualqblank{a} \dualqblank{b}} \dualqblank{c} = \qright{\dualqblank{b} \dualqblank{c}} \dualqblank{a} = \left( \qright{\dualqblank{b}} \dualqblank{a} \right) \dualqblank{c}
\end{gather}
for dual quaternions $\dualqblank{a}, \dualqblank{b}, \dualqblank{c} \in \dquats$.

Lastly, the dual quaternion conjugate can also be deconstructed using the dual conjugate matrix $\dualeyeconj$,
\begin{gather}
    \dualeyeconj = \blkdiag{\eyeconj, \eyeconj} \\
    \qconj{\dualqblank{a}} = \eyeconj \dualqblank{a}, \quad
    \qconj{\left( \dualqblank{a} \dualqblank{b} \right)} = \qconj{\dualqblank{b}} \qconj{\dualqblank{a}}.
\end{gather}

\subsubsection{Derivatives}
By rewriting dual quaternion operations in a matrix form, derivatives become simpler to calculate and can be written in a compact form that mirrors quaternion derivatives. Some common derivatives that will be used in the observability analysis in Section \ref{sec:obsv} are:
\begin{gather}
    \pd{\norm{\dualqblank{q}}^2}{\dualqblank{q}} = 2 \dualqblank{q}, \quad
    \pd{\dualqblank{a} \dualqblank{b}}{\dualqblank{a}} = \qright{\dualqblank{b}}, \quad
    \pd{\dualqblank{a} \dualqblank{b}}{\dualqblank{b}} = \qleft{\dualqblank{a}} \\
    \pd{\qconj{\dualqblank{a}} \dualqblank{b} \dualqblank{a}}{\dualqblank{a}} = \qleft{\qconj{\dualqblank{a}} \dualqblank{b}} + \qright{\dualqblank{b} \dualqblank{a}} \dualeyeconj.
    % \pd{\qreal{\qblank{q}}}{\dualqblank{q}} = \mtx{\eye{4} & 0} \\
    % \pd{\qdual{\qblank{q}}}{\dualqblank{q}} = \mtx{0 & \eye{4}} \\
\end{gather}

% --- modeling
\section{STATE-SPACE MODEL} \label{sec:model}
In this section, the dual quaternion relative motion and measurement models are presented. First, the relative state definition is introduced and then the relative rigid body dynamics are derived in the following subsection. The relative dynamics are derived by first discussing how to represent external forces and torques as a dual quaternion and then by using the rigid body dynamics with respect to an inertial frame to derive the relative dynamics. In the last subsection, the dual quaternion measurement model for a single fiducial marker relative pose measurement is presented.

\subsection{State Definition}
The dual quaternion relative motion model tracks the pose and velocities of the target coordinate frame, $\target$, with respect to the chaser/camera coordinate frame, $\camera$. The state, $x \in \reals{16}$, is composed of the dual pose and dual velocity quaternions of $\target$ with respect to $\camera$ in $\camera$ coordinates:
\begin{align}
    x &= \mtx{\dualq{\target}{\camera} \\ \dualomega{\target}{\camera}{\camera}}
    = \mtx{\q{\target}{\camera} \\ 
            \frac{1}{2} \qpos{\target}{\camera}{\camera} \q{\target}{\camera} \\
            \qomega{\target}{\camera}{\camera} \\
            \qvel{\target}{\camera}{\camera} + \qomega{\target}{\camera}{\camera} \cross \qpos{\camera}{\target}{\camera}}.
\end{align}

\subsection{Rigid Body Dynamics}
The vector dual quaternion force, $\dualforce{\camera} \in \dquatsv$, is the net external forces ($\force{\camera} \in \reals{3}$) and torques ($\torque{\camera} \in \reals{3}$) in $\cf{\camera}$ frame coordinates
\begin{gather}
    \dualforce{\camera} = \qforce{\camera} + \dualunit \qtorque{\camera} \\
    \qforce{\camera} = { \left[ 0, \force{\camera}\right] } \in \quatsv, \quad \qtorque{\camera} = { \left[ 0, \torque{\camera} \right] } \in \quatsv.
\end{gather}
The mass matrix, $\mass{\camera}$, contains the mass ($m$) and inertia ($\inertia \in \reals{3 \times 3}$) properties of the system in $\cf{\camera}$ frame coordinates
\begin{align}
    \mass{\camera} = \blkdiag{1, m \eye{3}, 1, \inertia}.
\end{align}
Note that calculating the external forces and torques may require knowledge of the pose of the camera and/or target with respect to an inertial frame. For a specific example of modeling dual quaternion forces and torques acting on a spacecraft please refer to \cite{Filipe2015-zc}.

% For the satellite relative motion problem, the pose of the camera or target with respect to an Earth-centered inertial frame, $\inertial$, must be known to calculate external forces and torques for the dynamics. In this work, we assume the camera's inertial pose, $\dualq{\camera}{\inertial}$, and velocities, $\dualomega{\camera}{\inertial}{\camera}$, are known, but in other applications, these quantities may be unknown and the target's inertial pose known (e.g., the target is a fixed docking station). Additionally, the mass and inertia properties of the camera and target must be known to convert the forces and torques to accelerations. 

The equations of motion for the camera with respect to the inertial frame, $\inertial$, in dual quaternion form are \cite{Tsiotras2020-od}
\begin{align}
    \ddualq{\camera}{\inertial} &= \frac{1}{2} \dualq{\camera}{\inertial} \dualomega{\camera}{\inertial}{\camera} \\
    \ddualomega{\camera}{\inertial}{\camera} &= \swap{ \left( \inv{(\mass{\camera})}  (\dualforce{\camera} - \dualomega{\camera}{\inertial}{\camera} \cross \mass{\camera} \swap{(\dualomega{\camera}{\inertial}{\camera})}) \right)}. \label{eq:ci_dynamics}
\end{align}
The target pose and velocity with respect to the inertial frame has the same form
\begin{align}
    \ddualq{\target}{\inertial} &= \frac{1}{2} \dualq{\target}{\inertial} \dualomega{\target}{\inertial}{\target} \\
    \ddualomega{\target}{\inertial}{\target} &= \swap{ \left( \inv{(\mass{\target})}  (\dualforce{\target} - \dualomega{\target}{\inertial}{\target} \cross \mass{\target} \swap{(\dualomega{\target}{\inertial}{\target})}) \right) }.
\end{align}

We will now derive the relative dynamics. First, use the fact that $\dualomega{\target}{\camera}{\camera} = \dualomega{\target}{\inertial}{\camera} - \dualomega{\camera}{\inertial}{\camera}$, and then differentiate with respect to time and apply the dual quaternion Transport Theorem \cite{Filipe2013-bj}:
\begin{align}
    \ddualomega{\target}{\camera}{\camera} &= \ddualomega{\target}{\inertial}{\camera} - \ddualomega{\camera}{\inertial}{\camera} \\
    \ddualomega{\target}{\inertial}{\camera} &= \dualq{\target}{\camera} \ddualomega{\target}{\inertial}{\target} \qconj{\dualq{\target}{\camera}} + \dualomega{\target}{\camera}{\camera} \cross \dualomega{\target}{\inertial}{\camera}.
    % \ddualomega{\target}{\inertial}{\camera} = \ddualomega{\target}{\camera}{\camera} + \ddualomega{\camera}{\inertial}{\camera} + \dualomega{\camera}{\inertial}{\camera} \cross \dualomega{\target}{\camera}{\camera} \\
    % \ddualomega{\camera}{\inertial}{\camera} = \ddualomega{\target}{\inertial}{\camera} - \ddualomega{\target}{\camera}{\camera} - \dualomega{\camera}{\inertial}{\camera} \cross \dualomega{\target}{\camera}{\camera}.
\end{align}
Now substitute $\ddualomega{\camera}{\inertial}{\camera}$ from \eqref{eq:ci_dynamics}
% \begin{gather}
%     \ddualomega{\target}{\inertial}{\camera} - \ddualomega{\target}{\camera}{\camera} - \dualomega{\camera}{\inertial}{\camera} \cross \dualomega{\target}{\camera}{\camera} = \\ \swap{\left( \inv{(\mass{\camera})}  (\dualforce{\camera} - \dualomega{\camera}{\inertial}{\camera} \cross \mass{\camera} \swap{(\dualomega{\camera}{\inertial}{\camera})}) \right)} \nonumber \\
%     \ddualomega{\target}{\camera}{\camera} = - \swap{\left( \inv{(\mass{\camera})}  (\dualforce{\camera} - \dualomega{\camera}{\inertial}{\camera} \cross \mass{\camera} \swap{(\dualomega{\camera}{\inertial}{\camera})}) \right)} \\ + \ddualomega{\target}{\inertial}{\camera} - \dualomega{\camera}{\inertial}{\camera} \cross \dualomega{\target}{\camera}{\camera} \nonumber
% \end{gather}
to obtain the complete relative pose equations of motion:
\begin{align}
    \ddualq{\target}{\camera} &= \frac{1}{2}  \dualomega{\target}{\camera}{\camera} \dualq{\target}{\camera} \label{eq:kinematics} \\
    % \ddualomega{\target}{\camera}{\camera} &= - \swap{\left( \inv{(\mass{\camera})}  (\dualforce{\camera} - \dualomega{\camera}{\inertial}{\camera} \cross \mass{\camera} \swap{(\dualomega{\camera}{\inertial}{\camera})}) \right)} \label{eq:dynamics}\\
    % & + \dualq{\target}{\camera} \ddualomega{\target}{\inertial}{\target} \qconj{\dualq{\target}{\camera}} + \dualomega{\target}{\camera}{\camera} \cross \dualomega{\target}{\inertial}{\camera}. \nonumber
    \ddualomega{\target}{\camera}{\camera} &= \dualq{\target}{\camera} \ddualomega{\target}{\inertial}{\target} \qconj{\dualq{\target}{\camera}} + \dualomega{\target}{\camera}{\camera} \cross \dualomega{\target}{\inertial}{\camera} \label{eq:dynamics} \\
    &- \swap{\left( \inv{(\mass{\camera})}  (\dualforce{\camera} - \dualomega{\camera}{\inertial}{\camera} \cross \mass{\camera} \swap{(\dualomega{\camera}{\inertial}{\camera})}) \right)} . \nonumber
\end{align}

\subsection{Measurement Model}
The measurement model is based on the AprilTag visual fiducial system \cite{Wang2016-ae}. The detection of a single fiducial marker gives the relative pose of the marker, $\marker$, with respect to the camera. The pose of the marker with respect to the target is assumed to be known and is represented as the dual pose
\begin{align}
    \dualq{\marker}{\target} &= \mtx{\q{\marker}{\target} \\ \frac{1}{2} \qpos{\marker}{\target}{\target} \q{\marker}{\target}}.
\end{align}
Then the measurement of a single marker is the dual pose quaternion, $\dualq{\marker}{\camera}$, where
\begin{align}
    y &= h(x) = \dualq{\marker}{\camera} = \dualq{\target}{\camera} \dualq{\marker}{\target}. \label{eq:meas}
\end{align}

% --- observability
\section{NONLINEAR OBSERVABILITY} \label{sec:nlobsv}
While there is a single definition of observability for linear systems, with nonlinear systems there are multiple degrees of observability, and we must define exactly what kind of observability is being considered. This section is a brief review of nonlinear observability and the Lie derivative approach for determining the observability of nonlinear systems. The definitions of the various classes of observability in this section are summarized from \cite{Hermann1977-rt, Powel2015, Zhou2008-me, Nijmeijer1990}.

Consider the nonlinear system, $\Sigma$, with motion and measurement models
\begin{gather}
    \Sigma: \quad \dot{x} = f(x, u), \quad
    y = h(x),
\end{gather}
where $x(t) \in \reals{n}$, \ $u(t) \in \mathcal{U} \subseteq \reals{m}$, and $\mathcal{U}$ is the set of permissible controls. Let the solution to the initial value problem for $\Sigma$ for $x(0) = x_0$ with the control input $u(t)$ be $x(t, x_0, u)$, and let $y(t,x_0, u) = h(x(t, x_0, u))$. 
% We will assume that $f, h \in C^\infty$, which is sufficient to guarantee the existence and uniqueness of a solution to $\Sigma$.

Points $x_0, x_1 \in \reals{n}$ are \emph{indistinguishable} if for every control $u \in \mathcal{U}$, then $y(t, x_0, u) = y(t, x_1, u)$ for all $t$. The system $\Sigma$ is \emph{weakly observable at $x_0$} if there exists an open neighborhood $U$ of $x_0$ such that if $x_1 \in U$ and $x_0$ and $x_1$ are indistinguishable, then $x_0 = x_1$. $\Sigma$ is \emph{weakly observable} if $\Sigma$ is weakly observable at all $x$.

The points $x_0$ and $x_1$ are \emph{$U$-indistinguishable} if for every control, $u \in \mathcal{U}$, with trajectories $x(t, x_0, u)$ and $x(t, x_1, u)$ that lie in $U \subseteq \reals{n}$ for $t \in \left[ 0, T \right]$, then $y(t, x_0, u) = y(t, x_1, u)$ for all $t \in \left[0, T \right]$. $\Sigma$ is \emph{locally weakly observable at $x_0$} if there exists an open neighborhood $U$ of $x_0$ such that for every open neighborhood $V \subset U$ of $x_0$, $x_0$ and $x_1$ $V$-indistinguishable implies that $x_0 = x_1$, and $\Sigma$ is \emph{locally weakly observable} if $\Sigma$ is locally weakly observable at all $x$.

Local weak observability is a stronger definition and implies weak observability. Weak observability at $x_0$ implies that $x_0$ can be eventually distinguished from its neighbors for some control, but may require traveling far away from the initial condition. Local weak observability implies that $x_0$ can be distinguished from its neighbors in finite time and space. While observability can be a global property, both definitions only consider observability with respect to some neighborhood of an initial point.

An analytical approach to testing the observability of a nonlinear system can be derived from differential geometry and provides a rank criterion for determining local weak observability. The zeroth through second-order \emph{Lie derivatives} of the function $h(x)$ with respect to a vector field $f(x)$ are
\begin{align}
    \lie{h}{f}{0} &= h(x) \\
    \lie{h}{f}{1} &= \grad \lie{h}{f}{0} \ f(x) = \grad h(x) \ f(x) \\
    \lie{h}{f}{2} &= \grad \lie{h}{f}{1} \ f(x) = \lie{\left( \lie{h}{f}{1} \right)}{f}{1}.
\end{align}
If $h(x)$ is a scalar function, then $\grad h(x)$ is the gradient expressed as a row vector. If $h(x)$ is a vector function, then $\grad h(x)$ is the Jacobian matrix.

Higher-order Lie derivatives have a similar form and can be written in terms of lower-order Lie derivatives. Note that if the system is control affine ($\dot{x} = f_0(x) + \sum_{i = 1}^m f_i(x) u_i$) or has multiple measurements, then it is possible to take mixed Lie derivatives with respect to $f_0(x)$ and $f_i(x)$, and with respect to each measurement independently. Since the system considered in this work does not have a control input and only one measurement, this special case of the Lie derivatives will not be presented here but additional information and examples can be found in \cite{Zhou2008-me, Mirzaei2008-kl}.

The observability matrix is defined with rows
\begin{gather}
    \obsv{} = \left\{ \grad \lie{h}{f}{n} \ | \ n \in \mathbb{N} \right\}
\end{gather}
and is used to determine the observability of the nonlinear system through the following rank condition.
\begin{obsvrank}
    If the observability matrix of the nonlinear system is full rank, then the system is locally weakly observable.
\end{obsvrank}
There is no systematic way to construct the observability matrix, however, in practice, taking sequential Lie derivatives along well-chosen combinations of the motion model functions typically yields promising results. If any combination of candidate Lie derivatives of arbitrary degree satisfies the rank criteria, the system is locally weakly observable.

\section{OBSERVABILITY ANALYSIS} \label{sec:obsv}
Let $\minstate \in \reals{13}$ be the reduced state representation of $x$:
\begin{align}
    \minstate = \tpose{\mtx{\q{\target}{\camera} & \pos{\target}{\camera}{\camera} & \omeg{\target}{\camera}{\camera} & \vel{\target}{\camera}{\camera}}}.
\end{align}
The dual quaternion state, $x$, is a convenient embedding for $\minstate$, but for decision-making and practical purposes, we are actually concerned with the observability of $\minstate$.
\begin{theorem}
    If $x$ is observable, then $\minstate$ is observable.
\end{theorem}
\begin{proof}
    There is a bijective mapping $x = g(\minstate)$ with inverse mapping $\minstate = g^{-1}(x)$. Therefore, if $x$ is observable then $\minstate$ can be uniquely determined through the inverse mapping.
\end{proof}

In the remainder of this section, a series of lemmas and corollaries regarding triangular and block triangular matrices are presented and are then used to prove observability.
% First, recall Sylvester's inequality:
% \begin{theorem} \label{thm:sylvester}
%     Sylvester's Rank Inequality \cite{cookbook} \\
%     If $A$ is $m \times n$ and $B$ is $n \times r$ then
%     \vspace{-0.2cm}
%     \begin{align}
%         \rank{A} + \rank{B} - n \leq \rank{AB}
%     \end{align}
% \end{theorem}
% \begin{sylvester}[\cite{cookbook}] \label{thm:sylvester}
%     If $A$ is $m \times n$ and $B$ is $n \times r$ then
%     \vspace{-0.2cm}
%     \begin{align}
%         \rank{A} + \rank{B} - n \leq \rank{AB}.
%     \end{align}
% \end{sylvester}
% \begin{corollary} \label{cor:squarerank}
%     If $A$ and $B$ are square matrices with full rank, then $AB$ is full rank.
% \end{corollary}
Triangular and block triangular matrices regularly appear when working with dual quaternions as in \eqref{eq:dqleftright}.
\begin{lemma}
     The eigenvalues of a triangular matrix are its diagonal entries \cite{Hefferon2020-answers}.
\end{lemma}
% \begin{proof}
%     The characteristic polynomial of a triangular matrix $A$ is:
%     \begin{align}
%         p_A(s) &= \det \left(sI - A \right) \\
%         &=  (s - a_{11}) (s - a_{22}) \dots (s - a_{nn})
%     \end{align}
%     The roots of the characteristic polynomial $\{a_{11}, a_{22}, \dots, a_{nn} \}$ are the eigenvalues of $A$.
% \end{proof}

\begin{corollary}
    A proper triangular matrix (non-zero entries on the diagonal) is full rank.
\end{corollary}

\begin{lemma} \label{lem:blockdiag}
    A block triangular square matrix is full rank if its square diagonal blocks are full rank \cite{blockrank}.
\end{lemma}
% \begin{proof}
%     Consider the block matrix $A$ with square diagonal blocks $B$ and $D$. The diagonal blocks can be decomposed using their SVDs as:
%     \begin{align}
%         A &= \mtx{B & 0 \\ C & D} \\
%         &= \mtx{U \Sigma V^* & 0 \\
%         C & W \Lambda Q^*} \\
%         &= \mtx{U & 0 \\ 0 & W} \mtx{\Sigma & 0 \\ W^* C V & \Lambda} \mtx{V^* & 0 \\ 0 & Q^*} \label{pf:blockdiag}
%     \end{align}
%     The matrices $U, W, V, Q$ are orthonormal and therefore the left and right matrices in \ref{pf:blockdiag} are full rank. The middle matrix in \ref{pf:blockdiag} is lower triangular and proper if $B$ and $D$ are full rank. Applying corollary \ref{cor:squarerank} yields that $A$ is full rank if $B$ and $D$ are full rank.
% \end{proof}

Lastly, we will prove that the rank of left and right quaternion and dual quaternion multiplication matrices are full rank.
\begin{lemma} \label{lem:qlr}
    The matrices $\qleft{\qblank{q}}$ and $\qright{\qblank{q}}$ are full rank if $q$ is a unit quaternion.
\end{lemma}
\begin{proof}
    If $\qblank{q}$ is a unit quaternion, then from the unit quaternion definition, its entries satisfy $q_0^2 + q_1^2 + q_2^2 + q_3^2 = 1$. The determinant and rank are related through the fact that a matrix $A$ is full rank if $\det(A) \neq 0$ \cite{Hefferon2020}. The determinant of $\qleft{\qblank{q}}$ is:
    \begin{align}
        \det \left( \qleft{\qblank{q}} \right) &= q_{0}^{4} + 2 q_{0}^{2} q_{1}^{2} + 2 q_{0}^{2} q_{2}^{2} + 2 q_{0}^{2} q_{3}^{2} + q_{1}^{4} + \\
        & 2 q_{1}^{2} q_{2}^{2} + 2 q_{1}^{2} q_{3}^{2} + q_{2}^{4} + 2 q_{2}^{2} q_{3}^{2} + q_{3}^{4} \nonumber \\
        &= (q_0^2 + q_1^2 + q_2^2 + q_3^2)^2 = 1.
    \end{align}
    Note that $\det \left( \qleft{\qblank{q}} \right) = \det \left( \qright{\qblank{q}} \right)$, and therefore $\qright{\qblank{q}}$ is also full rank.
\end{proof}

\begin{corollary} \label{cor:dualqlr}
    The matrices $\qleft{\dualqblank{q}}$ and $\qright{\dualqblank{q}}$ are full rank if $\dualqblank{q}$ is a unit dual quaternion.
\end{corollary}
\begin{proof}
    It was shown in Lemma \ref{lem:qlr} that the matrices $\qleft{\qblank{q}}$ and $\qright{\qblank{q}}$ are full rank. It then follows from Lemma \ref{lem:blockdiag} that $\qleft{\dualqblank{q}}$ and $\qright{\dualqblank{q}}$ are full rank because they are block triangular with diagonal blocks that are square and full rank.
\end{proof}

We will now prove the observability of the dual quaternion relative motion system \eqref{eq:kinematics} and \eqref{eq:dynamics}, with a single fiducial marker relative pose measurement \eqref{eq:meas} by deriving the necessary zeroth and first-order Lie derivatives and constructing the observability matrix.

\subsection*{Zeroth Order Lie Derivatives}
\vspace{-0.6cm}
\begin{align}
    \lie{h}{f}{0} &= h(x) = \dualq{\marker}{\camera} = \dualq{\target}{\camera} \dualq{\marker}{\target} \\
    \jacobian \lie{h}{f}{0} &= \mtx{\pd{\lie{h}{f}{0}}{\dualq{\target}{\camera}} & \pd{\lie{h}{f}{0}}{\dualomega{\target}{\camera}{\camera}}} = \mtx{\qright{\dualq{\marker}{\target}} & \zeros{8}{8}}
\end{align}

\subsection*{First Order Lie Derivatives}
\vspace{-0.6cm}
\begin{align}
    \lie{h}{f}{1} &= \jacobian \lie{h}{f}{0} \ f(x) \\
    &= \mtx{\qright{\dualq{\marker}{\target}} & \zeros{8}{8}} \mtx{\frac{1}{2}  \dualomega{\target}{\camera}{\camera} \dualq{\target}{\camera} \\ \ddualomega{\target}{\camera}{\camera}} \\
    &= \frac{1}{2} \qright{\dualq{\marker}{\target}} \left(\dualomega{\target}{\camera}{\camera} \dualq{\target}{\camera}\right) \\
    &= \frac{1}{2} \dualomega{\target}{\camera}{\camera} \dualq{\target}{\camera} \dualq{\marker}{\target} \\
    \jacobian \lie{h}{f}{1} &= \mtx{\frac{1}{2} \qleft{\dualomega{\target}{\camera}{\camera}} \qright{\dualq{\marker}{\target}} & \frac{1}{2} \qright{\dualq{\target}{\camera} \dualq{\marker}{\target}}}
\end{align}

The Lie derivatives are concatenated to construct the observability matrix:
\begin{align}
    \obsv{} &= \mtx{\jacobian \lie{h}{f}{0} \\ \jacobian \lie{h}{f}{1}} \in \reals{16 \times 16} \\ 
    &= \mtx{\qright{\dualq{\marker}{\target}} & \zeros{8}{8} \\
    \frac{1}{2} \qleft{\dualomega{\target}{\camera}{\camera}} \qright{\dualq{\marker}{\target}} & \frac{1}{2} \qright{\dualq{\target}{\camera} \dualq{\marker}{\target}}}.
\end{align}
\begin{theorem}
    The observability matrix, $\obsv{}$, is full rank, and the system is observable with a single fiducial marker.
\end{theorem}
\begin{proof}
    The product $\dualq{\target}{\camera} \dualq{\marker}{\target}$ is a unit dual quaternion, and from Corollary \ref{cor:dualqlr} the matrix $\frac{1}{2} \qright{\dualq{\target}{\camera} \dualq{\marker}{\target}}$ is therefore full rank. The matrix $\qright{\dualq{\marker}{\target}}$ is also full rank from Corollary \ref{cor:dualqlr}. By Lemma \ref{lem:blockdiag}, $\obsv{}$ is full rank because it is block diagonal, and its diagonal blocks are square and full rank.
\end{proof}

A noteworthy result of this derivation is that the dynamics are not needed for the system to be observable; the kinematics are sufficient. This result is expected because a time history of the relative pose measurement acquired by a single fiducial marker provides full-state information without a need for the dynamics. While the motivating scenario for this work was the relative motion between satellites, the results hold for any dynamics model. Additionally, there are no inertial terms that appear in the observability matrix, which aligns with the expected result that only relative pose information between the marker, camera, and target is necessary to fully and uniquely reconstruct the relative state.

% --- conclusion
\section{CONCLUSIONS}
Dual quaternions provide a promising framework for rigid body pose modeling and analysis. In this paper, we derived and applied properties of the dual quaternion Jacobian matrices to yield an observability analysis that required reasoning about simple block and triangular matrices. In future work, we will use results from this paper to investigate the observability of the same relative motion model in the presence of relative position and relative angles-only measurement models. The observability conditions found in future work will be leveraged in simulation and hardware experiments to demonstrate analytical results.

% \addtolength{\textheight}{-12cm}   % This command serves to balance the column lengths
%                                   on the last page of the document manually. It shortens
%                                   the textheight of the last page by a suitable amount.
%                                   This command does not take effect until the next page
%                                   so it should come on the page before the last. Make
%                                   sure that you do not shorten the textheight too much.

% --- appendix
% \section*{APPENDIX}
% Appendixes should appear before the acknowledgment.

% --- acknowledgement
\section*{ACKNOWLEDGMENT}
The authors thank Nathan Powel for his guidance and insightful comments.

% --- references
\bibliographystyle{IEEEtran}
\bibliography{IEEEabrv, references}

\end{document}